%% file: adaptive_main.tex
\documentclass[final]{siamltex}
\usepackage{amsmath,bm,amssymb}
\usepackage{amsfonts}
\usepackage{graphicx}

\usepackage[T1]{fontenc}
\usepackage[ansinew]{inputenc}
\usepackage{lmodern}
\usepackage{subfigure}

\usepackage{todonotes}

\usepackage{ifpdf}
\ifpdf
  \usepackage{epstopdf} 
\fi

\usepackage[letterpaper]{geometry}
\usepackage[colorlinks=true, citecolor=blue, filecolor=black,linkcolor=red, urlcolor=blue, hypertexnames=false]{hyperref}
\usepackage{algorithm}
\usepackage{algpseudocode}

\usepackage{todonotes}

\DeclareMathAlphabet\mathpzc{OT1}{pzc}{m}{it}
\def\E{{\mathbb E}}
\def\P{{\mathbb P}}
\def\dis{{\mathcal{D}}}

\newcommand\partialderiv[3][]{\frac{\partial^{#1}#2}{\partial{#3}^{#1}}}

\let\trueiiint=\iiint
\def\iiint{\mathop{\textstyle\trueiiint}\limits}
\def\intinfty{\int\limits_{\!\!-\infty\,\,}^{\,\,\infty\!\!}\kern-0.0em}
\def\iintinfty{\mathop{\int\!\!\int}\limits_{\!\!-\infty\,\,}^{\,\,\infty\!\!}\kern-0.0em}
\def\iiintinfty{\mathop{\int\!\!\int\!\!\int}\limits_{\!\!-\infty\,\,}^{\,\,\infty\!\!}\kern-0.0em}
\def\Real{{\mathbb R}}

\let\<=\langle
\let\>=\rangle
\let\^=\hat
\def\~#1{{\-ox{\sf#1}}}

\let\-=\mathbf

\def\bolde{{\boldsymbol\epsilon}}

\newtheorem{assumption}[theorem]{Assumption}

\def\circ{\ifmmode\mathchar"220E\else$\mathchar"220E$\fi}
\def\@#1{{\cal #1}}

\makeatletter
\newlength{\continueindent}
\renewenvironment{algorithmic}[1][0]%
   {%
   \edef\ALG@numberfreq{#1}%
   \def\@currentlabel{\theALG@line}%
   \setcounter{ALG@line}{0}%
   \setcounter{ALG@rem}{0}%
   \let\\\algbreak%
   \expandafter\edef\csname ALG@currentblock@\theALG@nested\endcsname{0}%
   \expandafter\let\csname ALG@currentlifetime@\theALG@nested\endcsname\relax%
   \begin{list}%
      {\ALG@step}%
      {%
      \rightmargin\z@%
      \itemsep\z@ \itemindent\z@ \listparindent2em%
      \partopsep\z@ \parskip\z@ \parsep\z@%
      \labelsep 0.5em \topsep 0.2em
      \ifthenelse{\equal{#1}{0}}%
         {\labelwidth 0.5em}%
         {\labelwidth 1.2em}%
       \leftmargin\labelwidth \addtolength{\leftmargin}{\labelsep}
      \ALG@tlm\z@%
      }%
      \parshape 2 \leftmargin \linewidth \continueindent \dimexpr\linewidth-\continueindent\relax
   \setcounter{ALG@nested}{0}%
   \ALG@beginalgorithmic%
   }%
   {
   \ALG@closeloops%
   \expandafter\ifnum\csname ALG@currentblock@\theALG@nested\endcsname=0\relax%
   \else%
      \PackageError{algorithmicx}{Some blocks are not closed!!!}{}%
   \fi%
   \ALG@endalgorithmic%
   \end{list}%
   }%
\makeatother

\title{Adaptive construction of surrogates for the Bayesian solution
  of inverse problems\thanks{The work was supported by the United
    States Department of Energy, Office of Advanced Scientific
    Computing Research (ASCR), under grant DE-SC0002517.}}

\author{Jinglai Li\footnotemark[2]
        \and Youssef M.\ Marzouk\footnotemark[3]}

\begin{document}

\maketitle

\renewcommand{\thefootnote}{\fnsymbol{footnote}}

\footnotetext[2]{Institute of Natural Sciences, Shanghai Jiaotong
  University, Shanghai 200240, China, \texttt{jinglaili@sjtu.edu.cn}.}

\footnotetext[3]{Department of Aeronautics and
Astronautics, Massachusetts Institute of Technology, Cambridge, MA
02139, USA, \texttt{ymarz@mit.edu}.}

\renewcommand{\thefootnote}{\arabic{footnote}}

\begin{abstract}
  The Bayesian approach to inverse problems typically relies on
  posterior sampling approaches, such as Markov chain Monte Carlo, for
  which the generation of each sample requires one or more evaluations
  of the parameter-to-observable map or forward model. When these
  evaluations are computationally intensive, approximations of the
  forward model are essential to accelerating sample-based inference.
  Yet the construction of globally accurate approximations for
  nonlinear forward models can be computationally prohibitive and in
  fact unnecessary, as the posterior distribution typically
  concentrates on a small fraction of the support of the prior
  distribution.
  %
  We present a new approach that uses stochastic optimization to
  construct polynomial approximations over a sequence of measures
  adaptively determined from the data, eventually concentrating on
  the posterior distribution.
  The approach yields substantial gains in efficiency and accuracy
  over prior-based surrogates,
  as demonstrated via application to inverse problems in partial
  differential equations.

\end{abstract}

\pagestyle{myheadings}
\thispagestyle{plain}
\markboth{\MakeUppercase{J.\ Li and Y.\ M.\ Marzouk}}{\MakeUppercase{Adaptive construction of surrogates for inverse problems}}

\input{intro}
\input{review}
\input{method}

\input{example}
\input{conclusion}

\input{appendix}
\bibliographystyle{siam}
\bibliography{jinglaibib,localbib}


\end{document}

%% file: intro.tex
\section{Introduction}\label{s:intro}
In many science and engineering problems, parameters of interest cannot be observed directly; instead, they must be estimated from indirect observations. In these situations, one can usually appeal to a \textit{forward model} mapping the parameters of interest to some quantities that can be measured. The corresponding \textit{inverse problem} then involves inferring the unknown parameters from a set of observations~\cite{Evans:2002}.

Inverse problems arise in a host of applications, ranging from the geosciences \cite{tarantola2005inverse, Mosegaard:2002, Malinverno:2002, cui2011bayesian} to chemical kinetics \cite{Najm:2008a} and far beyond. In these applications, data are inevitably noisy and often limited in number.  The Bayesian approach to inverse problems~\cite{kaipio2005statistical,stuart2010inverse,tarantola2005inverse,wang2005hierarchical} provides a foundation for inference from noisy and incomplete data, a natural mechanism for incorporating physical constraints and heterogeneous sources of information, and a quantitative assessment of uncertainty in the inverse solution. Indeed, the Bayesian approach casts the inverse solution as a posterior probability distribution over the model parameters or inputs. Though conceptually straightforward, this setting presents challenges in practice. The posterior distributions are typically not of analytical form or from a standard parametric family;
characterizing them exactly requires sampling approaches such as Markov chain Monte Carlo~\cite{besag1995bayesian,chen2000monte,tierney1994markov,MCMChandbook}. These methods entail repeated solutions of the forward model. When the forward model is computationally intensive, e.g., specified by partial differential equations (PDEs), a direct approach to Bayesian inference becomes computationally prohibitive.

One way to accelerate inference is to construct computationally inexpensive approximations of the forward model and to use these approximations as surrogates in the sampling procedure. Many approximation methods have been successfully employed in this context, ranging from projection-based model reduction~\cite{Wang200515,Nguyen_SantaFE08,Galbally09} to Gaussian process regression~\cite{williams2006combining,RSSB:RSSB294} to parametric regression~\cite{balakrishnan2003uncertainty}. Our previous work \cite{Marzouk2009,Marzouk2007,Marzouk2009a} has developed surrogate models by using stochastic spectral methods \cite{ghanem2003stochastic,xiu2010numerical} to propagate prior uncertainty from the inversion parameters to the forward model outputs. The result is a forward model approximation that converges in the \textit{prior-weighted} $L^2$ sense. Theoretical analysis shows that if the forward model approximation converges at certain rate in this prior-weighted norm, then (under certain assumptions) the posterior distribution generated by the approximation converges to the true posterior at the same rate~\cite{stuart2010inverse,Marzouk2009a,lucor}. Constructing a sufficiently accurate surrogate model over the support of the prior distribution, however, may not be possible in many practical problems. When the dimension of the parameters is large and the forward model is highly nonlinear, constructing such a ``globally accurate'' surrogate can in fact be a formidable task.

The inverse problem fortunately has more structure than the prior-based uncertainty propagation problem. Since the posterior distribution reflects some information gain relative to the prior distribution, it often concentrates on a much smaller portion of the parameter space. In this paper, we will propose that:\ (i) it can therefore be more efficient to construct a surrogate that only maintains high accuracy in the regions of appreciable posterior measure; and (ii) this ``localized'' surrogate can enable accurate posterior sampling and accurate computation of posterior expectations.

A natural question to ask, then, is how to build a surrogate in the important region of the posterior distribution before actually characterizing the posterior? Inspired by the cross-entropy (CE) method~\cite{de2005tutorial,rubinstein2004cross} for rare event simulation, 
we propose an adaptive algorithm to find a distribution that is ``close'' to the posterior in the sense of Kullback-Leibler divergence and to build a local surrogate with respect to this approximating distribution. Candidate distributions are chosen from a simple parameterized family, and the algorithm minimizes the Kullback-Leibler divergence of the candidate distribution from the posterior using a sequence of intermediate steps, where the optimization in each step is accelerated through the use of locally-constructed surrogates. We demonstrate with numerical examples that the total computational cost of our method is much lower than the cost of building a globally accurate surrogate of comparable or even lesser accuracy. Moreover, we show that the final approximating distribution can provide an excellent proposal for Markov chain Monte Carlo sampling, in some cases exceeding the performance of adaptive random-walk samplers. This aspect of our methodology has links to previous work in adaptive independence samplers~\cite{keith2008adaptive}.

The remainder of this article is organized as follows. In Section~\ref{s:review} we briefly review the Bayesian formulation of inverse problems and previous work using polynomial chaos surrogates to accelerate inference. In Section~\ref{s:method}, we present our new adaptive method for the construction of local surrogates, with a detailed discussion of the algorithm and an analysis of its convergence properties. Section~\ref{s:example} provides several numerical demonstrations, and Section~\ref{s:conclusion} concludes with further discussion and suggestions for future work.


%% file: review.tex
\section{Bayesian inference and polynomial chaos
  surrogates}\label{s:review}

Let $\-y\in \Real^{n_y}$ be the vector of parameters of interest
and $\-d\in\Real^{n_d}$ be a vector of observed data. In the Bayesian
formulation, prior information about $\-y$ is encoded in the prior
probability density $\pi(\-y)$ and related to the posterior
probability density $\pi(\-y|\-d)$ through Bayes' rule:
\begin{equation}
\pi(\-y|\-d) = \frac{\pi(\-d|\-y)\pi(\-y)}{\int \pi(\-d|\-y)\pi(\-y) d\-y\,},
\label{e:bayesrule}
\end{equation}
where $\pi(\-d|\-y)$ is the likelihood function. (In what
follows, we will restrict our attention to finite-dimensional
parameters $\-y$ and assume that all random variables have densities
with respect to Lebesgue measure.)
The likelihood function incorporates both the data and the forward model. In the context of an inverse problem, the likelihood usually results from some combination of a deterministic forward model $\-G: \Real^{n_y} \rightarrow \Real^{n_d}$ and a statistical model for the measurement noise and model error. For example, assuming additive measurement noise $\bolde$, we have
\begin{equation}
\-d = \-G(\-y)+\bolde \, .
\end{equation}
If the probability density of $\bolde$ is given by $\pi_\epsilon(\bolde)$, then the likelihood function becomes
\begin{equation}
\pi(\-d|\-y) = \pi_\epsilon(\-d-\-G(\-y))\ .\label{e:lh}
\end{equation}
For conciseness, we define $\pi^*(\-y) := \pi(\-y|\-d)$ and
$L(\-G) := \pi_\epsilon(\-d - \-G)$ and rewrite Bayes' rule \eqref{e:bayesrule} as
\begin{equation}
\pi^*(\-y) = \frac{L(\-G(\-y))\pi(\-y)}{I},
\label{e:bayesrule2}
\end{equation}
where 
\begin{equation}
I:=\int L(\-G(\-y))\pi(\-y) d\-y\,\label{e:I}
\end{equation}
is the posterior normalizing constant or evidence. In practice, no
closed form analytical expression for $\pi^*(\-y)$ exists, and any
posterior moments or expectations must be estimated via sampling
methods such as Markov chain Monte Carlo (MCMC), which require many
evaluations of the forward model.

If the forward model has a smooth dependence on its input parameters
$\-y$, then using stochastic spectral methods to accelerate this
computation is relatively straightforward. As described
in~\cite{Marzouk2009,Marzouk2007,Marzouk2009a,frangos:2010:srm}, the essential
idea behind existing methods is to construct a stochastic forward
problem whose solution approximates the deterministic forward model
$\-G(\-y)$ over the support of the prior $\pi(\-y)$. More precisely,
we seek a polynomial approximation $\widetilde{\-G}_N(\-y)$ that
converges to $\-G(\-y)$ in the prior-weighted $L^2$ norm. Informally,
this procedure ``propagates'' prior uncertainty through the forward
model and yields a computationally inexpensive surrogate that can
replace $\-G(\-y)$ in, e.g., MCMC simulations.

For simplicity, we assume prior independence of the input parameters, namely,
\begin{equation*}
\pi(\-y)= \prod_{j=1}^{n_y} \pi_j(y_j)\,.
\end{equation*}
(This assumption can be loosened when necessary; see, for example,
discussions in \cite{Babuska2010,soize:395}.) Since $\-G(\-y)$ is
multi-dimensional, we construct a polynomial chaos (PC) expansion for
each component of the model output. Suppose that $g(\-y)$ is a
component of $\-G(\-y)$; then its $N$-th order PC expansion is
\begin{eqnarray}
  g_N( \-y ) = 
  \sum_{|\-i|\leq N} a_{\-i}
  \Psi_{\-i}(\-y),
  \label{e:PCEForm2}
\end{eqnarray}
where $\-i:=\left(i_1,i_2,\ldots,i_{n_y}\right)$ is a
 multi-index with $|\-i| := i_1+i_2+\ldots+i_{n_y}$,
$a_{\-i}$ are the expansion
coefficients, and $\Psi_{\-i}$ are orthogonal polynomial basis functions,
defined as
\begin{eqnarray}
  \Psi_{\-i}(\-y) = \prod_{j=1}^{n_y}
  \psi_{i_j}(y_j) \, . 
\end{eqnarray}
Here $\psi_{i_j}(y_j)$ is the univariate polynomial of degree $i_j$, from a system satisfying orthogonality with respect to $\pi_j$: 
\begin{eqnarray}
  \E_{j}\left[\psi_{i}\psi_{i^\prime}\right] = \int
  \psi_i\left(y_j\right)\psi_{i^\prime}\left(y_j\right)\pi_j\left(y_j\right)\,dy_j =
  \delta_{i,i^\prime},
  \label{e:orthogonality}
\end{eqnarray}
where we assume that the polynomials have been properly normalized.
It follows that $\Psi_{\-i}(\-y)$ are $n_y$-variate orthonormal polynomials satisfying
\begin{eqnarray}
  \E\left[\Psi_{\-i}(\-y)\Psi_{\-i^\prime}(\-y)\right] = \int
  \Psi_{\-i}\left(\-y\right)\Psi_{\-i^\prime}\left(\-y\right)\pi(\-y)\,d\-y =
  \delta_{\-i,\-i^\prime}\,,
  \label{e:orth}
\end{eqnarray}
where $\delta_{\-i,\-i^\prime}=\prod_{j=1}^{n_y}\delta_{i_j,i^\prime_j}$.

Because of the orthogonality condition \eqref{e:orthogonality}, the
distribution over which we are constructing the polynomial
approximation---namely each prior distribution $\pi_j(y_j)$---determines
the polynomial type. For example, Hermite polynomials are associated
with the Gaussian distribution, Jacobi polynomials with the beta
distribution, and Laguerre polynomials with the gamma
distribution. For a detailed discussion of these correspondences and
their resulting computational efficiencies, see~\cite{Xiu2002a}. For
PC expansions corresponding to non-standard distributions, see
\cite{Wan2006,arnst2012measure}.
Note also that in the equations above, we have restricted our attention to
total-order polynomial expansions, i.e., $|\-i|\leq N$. This choice is
merely for simplicity of exposition; in practice, one may choose any
admissible multi-index set $\mathcal{J} \ni \-i$ to define the
polynomial chaos expansion in \eqref{e:PCEForm2}.

The key computational task in constructing these polynomial
approximations is the evaluation of the expansion coefficients
$\-a_{\-i}$. Broadly speaking, there are two classes of methods for
doing this: intrusive (e.g., stochastic Galerkin) and non-intrusive
(e.g., interpolation or pseudospectral approximation). In this paper,
we will follow \cite{Marzouk2009a} and use a non-intrusive method to
compute the coefficients.  The main advantage of a non-intrusive
approach is that it only requires a finite number of deterministic
simulations of the forward model, rather than a reformulation of the
underlying equations. Using the orthogonality
relation~\eqref{e:orth}, the expansion coefficients are given by
\begin{equation}
a_{\-i}= \E[g(\-y)\Psi_{\-i}(\-y)\,] ,
\end{equation}
and thus $a_{\-i}$ can be estimated by numerical quadrature
\begin{equation}
\tilde{a}_{\-i} = \sum_{m=1}^M g(\-y^{(m)})\,\Psi_{\-i}(\-y^{(m)})\,w^{(m)},
\label{e:quad}
\end{equation}
where $\-y^{(m)}$ are a set of quadrature nodes and $w^{(m)}$ are the
associated weights for $m=1,\ldots,M$. Tensor product quadrature rules
are a natural choice, but for $n_y \geq 2$, using sparse quadrature
rules to select the model evaluation points can be vastly more efficient.
Care must be taken to avoid significant aliasing errors when using
sparse quadrature directly in \eqref{e:quad}, however. Indeed, it is
advantageous to recast the approximation as a Smolyak sum of
constituent full-tensor polynomial approximations, each associated
with a tensor-product quadrature rule that is appropriate to its
polynomials~\cite{constantine:2012:spa}. This type of
approximation may be constructed adaptively, thus taking advantage of
weak coupling and anisotropy in the dependence of $\-G$ on $\-y$. More
details can be found in~\cite{conrad:2012:asp}.


%% file: method.tex
\section{Adaptive surrogate construction}\label{s:method}

The polynomial chaos surrogates described in the previous section are
constructed to ensure accuracy with respect to the prior; that is,
they converge to the true forward model $\-G(\-y)$ in the $L^2_\pi$
sense, where $\pi$ is the prior density on $\-y$. In many inference
problems, however, the posterior is concentrated in a very small
portion of the entire prior support. In this situation, it can be much
more efficient to build surrogates only over the important region of
the posterior. (Consider, for example, a forward model output that
varies nonlinearly with $\-y$ over the support of the prior. Focusing
onto a much smaller range of $\-y$ reduces the degree of nonlinearity;
in the extreme case, if the posterior is sufficiently concentrated and
the model output is continuous, then even a linear surrogate could be
sufficient.) In this section we present a general method for
constructing posterior-focused surrogates.

\subsection{Minimizing cross entropy}


The main idea of our method is to build a polynomial chaos surrogate
over a probability distribution that is ``close'' to the posterior in
the sense of Kullback-Leibler (K-L) divergence. Specifically, we seek
a distribution with density $p(\-y)$ that minimizes the K-L divergence
from $\pi^\ast(\-y)$ to $p$:\footnote{Note that the K-L divergence is
  not symmetric; $\dis_\mathrm{KL}(\pi^\ast \| \, p) \neq
  \dis_\mathrm{KL}( p \| \, \pi^\ast)$. Minimizing the K-L divergence
  from $\pi^\ast$ to $p$ as in \eqref{e:ce} tends to yield a $p$
  that is broader than $\pi^\ast$, while minimizing K-L divergence in
  the opposite direction can lead to a more compact approximating
  distribution, for instance, one that concentrates on a single mode
  of $\pi^\ast$ \cite{mackay:2006:eoi, turner2008two}. The former behavior is far more
  desirable in the present context, as we seek a surrogate that
  encompasses the entire posterior.}
\begin{equation}
\dis_\mathrm{KL}(\pi^\ast \| \, p) = \int \pi^\ast(\-y)\ln \frac{\pi^\ast(\-y)}{p(\-y)}\,d\-y
=\int \pi^\ast(\-y)\ln{\pi^\ast(\-y)} \, d\-y - \int \pi^\ast(\-y)\ln{p(\-y)}\,d\-y\,.
\label{e:ce}
\end{equation}
Interestingly, one can minimize \eqref{e:ce} without exact knowledge
of the posterior distribution $\pi^\ast(\-y)$.  Since the first
integral on the right-hand side of \eqref{e:ce} is independent of $p$,
minimizing $\dis_\mathrm{KL}(\pi^\ast \| \, p)$ is equivalent to maximizing
\begin{equation*}
\int \pi^\ast(\-y)\ln{p(\-y)}\,d\-y\,=
\int {\frac{L(\-G(\-y))\pi(\-y)}I\ln p(\-y) \,d\-y} \, .
\end{equation*}
Moreover, since $I$ is a constant (for fixed data), one simply needs
to maximize $\int {L(\-G(\-y))\ln p(\-y)}\pi(\-y)\,d\-y$.

In practice, one selects the candidate distributions $p$ from a
parameterized family $\mathcal{P}_V = \{p(\-y;\-v)\}_{\-v\in V}$, where $\-v$ is a
vector of parameters (called ``reference parameters'' in the cross-entropy method for
rare event simulation) and $V$ is the corresponding parameter space.
Thus the desired distribution can be found by solving the optimization problem
\begin{equation}
\max_{\-v\in V} D(\-v) = \max_{\-v\in V}\int {L(\-G(\-y))\ln
  p(\-y;\-v) \pi(\-y)}\,d\-y\,.
\label{e:maxD}
\end{equation}

\subsection{Adaptive algorithm}

In this section, we propose an adaptive algorithm to solve the optimization
problem above. The algorithm has three main ingredients: sequential
importance sampling, a tempering procedure, and localized surrogate
models. In particular, we construct a sequence of intermediate
optimization problems that converge to the original one in
\eqref{e:maxD}, guided by a tempering parameter. In each intermediate
problem, we evaluate the objective using importance sampling and we
build a local surrogate to replace expensive evaluations of the
likelihood function.

We begin by recalling the essentials of importance sampling
\cite{robert:2004:mcs}. Importance sampling (IS) simulates from a biasing
distribution that may be different than the original distribution or
the true distribution of interest, but corrects for this mismatch by
weighing the samples with an appropriate ratio of densities. By
focusing samples on regions where an integrand is large, for example,
IS can reduce the variance of a Monte Carlo estimator of an integral
\cite{LiuMC:2001}. In the context of \eqref{e:maxD}, a na\"{i}ve Monte Carlo
estimator might sample from the prior $\pi$, but this estimator will
typically have extremely high variance: when the likelihood function
is large only on a small fraction of the prior samples, most of the
terms in the resulting Monte Carlo sum will be near zero. Instead, we
sample from a biasing distribution $q(\-y)$ and thus rewrite $D$ as:
\begin{subequations}
\begin{equation}
D(\-v) = \int {L(\-G(\-y))\ln p(\-y;\-v) \frac{\pi(\-y)}{q(\-y)}\,q(\-y)}\,d\-y\,,
\label{e:Dis}
\end{equation}
and obtain an unbiased IS estimator of $D$:
\begin{equation}
\hat{D}(\-v)
:=  \frac1M\sum_{m=1}^M{L \left ( \-G(\-y^{(m)}) \right) \, l(\-y^{(m)}}) \, \ln p(\-y^{(m)};\-v),
\label{e:Dmcis}
\end{equation}
\end{subequations}
where $l(\-y) := \pi(\-y)/q(\-y)$ is the density ratio or
\textit{weight} function and the samples $\{\-y^{(m)}\}$ in
\eqref{e:Dmcis} are drawn independently from $q(\-y)$.

Next we introduce a tempering parameter $\lambda$, putting
\begin{equation*}
L(\-y;\lambda) := [L(\-y)]^{\frac1\lambda},
\end{equation*}
and defining $\pi^\ast(\-y;\lambda)$ as the posterior density
associated with the tempered likelihood $L(\-y;\lambda)$.  Here
$L(\-y;\lambda)$ and $\pi^\ast(\-y;\lambda)$ revert to the original
likelihood function and posterior density, respectively, for
$\lambda=1$. The algorithm detailed below will ensure that
$\lambda=1$ is reached within a finite number of steps.

The essential idea of the algorithm is to construct a sequence of
biasing distributions $\left ( p(\-y;\-v_k) \right )$, where each $p$
is drawn from the parameterized family $\mathcal{P}$. Each biasing
distribution has two roles: first, to promote variance reduction via
importance sampling; and second, to serve as the input distribution
for constructing a local surrogate model $\widetilde{\-G}_k(\-y)$. The
final biasing distribution, found by solving the optimization problem
when $\lambda=1$, is the sought-after ``best approximation'' to the
posterior described above.


\begin{algorithm}
\caption{Adaptive algorithm}
\label{alg:adap}
\begin{algorithmic}[1]
\State \textbf{Input:} probability fraction $\rho \in (0,1)$, likelihood
function level $\gamma > 0$, minimum step size $\delta > 0$; initial biasing
distribution parameters $\mathbf{v}_0$, IS sample size $M$
\State \textbf{Initialize:} $k=0$, $\lambda_0 = \infty$
\While {$\lambda_k > 1$}

\State Construct $\widetilde{\mathbf{G}}_k(\mathbf{y})$, the polynomial chaos
approximation of $\mathbf{G}_k(\mathbf{y})$ with respect to $p(\mathbf{y};\mathbf{v}_k)$
\vspace*{2pt}
\State Draw $M$ samples $\{\mathbf{y}^{(1)}$, \ldots, $\mathbf{y}^{(M)}\}$ according to $p(\mathbf{y};\mathbf{v}_k)$.
\vspace*{2pt}
\State Compute $\lambda_{k+1}$ such that the largest $100\rho\%$
percent of the resulting likelihood function values
$\{L ( \widetilde{\mathbf{G}}_k(\mathbf{y}^{(m)}; \lambda_{k+1}) )\}_{m=1}^M$ are larger than
$\gamma$.
\vspace*{2pt}
\If {$\lambda_{k+1} > \lambda_k-\delta$}
\State $\lambda_{k+1} \gets \lambda_{k} - \delta$
\EndIf
\If {$\lambda_{k+1} < 1$}
\State $\lambda_{k+1} \gets 1$
\EndIf
\vspace*{2pt}
\State Solve the optimization problem: \label{line:opt}
\begin{equation*}
\mathbf{v}_{k+1} = \arg \max_{\mathbf{v}\in V} \hat{D}_{k+1}(\mathbf{v}) = \arg \max_{\mathbf{\mathbf{v}}\in V} \frac1M \sum^M_{n=1}
{L(\widetilde{\mathbf{G}}_k(\mathbf{y}^{(m)});\lambda_{k+1}) \, \ln p(\mathbf{y}^{(m)};\mathbf{v})
  \, l_{k}(\mathbf{y}^{(m)})}
\end{equation*}
where $l_{k}(\mathbf{y}) := \pi(\mathbf{y})/p(\mathbf{y};\mathbf{v}_k)$. 
\vspace*{2pt}
\State $k  \gets  k + 1$
\EndWhile
\end{algorithmic}
\end{algorithm}

Steps of the algorithm are detailed in Algorithm~\ref{alg:adap}. The
basic idea behind the iterations is the following: since
$p(\-y;\-v_k)$, the biasing distribution obtained at step $k$, is
close to $\pi^\ast(\-y;\lambda_k)$, then in the next step, if one can
choose $\lambda_{k+1}$ such that $\pi^\ast(\-y;\-v_{k})$ and
$\pi^\ast(\-y;\-v_{k+1})$ are close, the forward model surrogate and
the IS estimator based on $p(\-y;\-v_k)$ should be effective for the
optimization problem at step $k+1$. A more formal discussion of the
convergence properties of the algorithm is given in
Section~\ref{s:convergence}.
Note that for reasons of computational efficiency, we want the
distributions over which we construct the surrogates to remain
relatively localized. Thus a good choice for $\-v_0$ would keep the
variance of $p(\-y; \-v_0)$ relatively small and perhaps center it at
the prior mean. The value of $\lambda$ will automatically adjust to
the choice of initial biasing distribution.

An important part of the algorithm is the choice of a parameterized
family of distributions $\mathcal{P}_V$. The distributions should be
flexible yet easy to sample, and straightforward to use as an input
distribution for the construction of polynomial surrogates. To this
end, multivariate normal distributions are a convenient choice. Not
only do they suggest the use of the well-studied Gauss-Hermite
polynomial chaos expansion for the forward model, but they also make
it possible to solve the optimization problem in step~\ref{line:opt}
of Algorithm~\ref{alg:adap} analytically. For example, let
$p(\-y;\-v)$ be an uncorrelated multivariate Gaussian:
\begin{equation}
p(\-y;\-v)= \prod_{j=1}^{n_y} p_j(y_j)\, , \qquad p_j(y_j) = \frac1{\sqrt{2\pi}\sigma_j} \exp\left(-\frac{(y_j-\mu_j)^2}{2\sigma_j^2}\right)\,,
\label{e:normal}
\end{equation}
where the reference parameters are now
$\-v=(\mu_1,\ldots,\mu_{n_y},\sigma_1,\ldots,\sigma_{n_y})$.  Assuming
that $\hat{D}$ in step~\ref{line:opt} of Algorithm~\ref{alg:adap} is
convex and differentiable with respect to $\-v$, we obtain the
solution to $\max_{\-v} \hat{D}$ by solving
\begin{equation}
\nabla_{\-v} \hat{D} = 0 \, .\label{e:gradD}
\end{equation}
Substituting \eqref{e:normal} into \eqref{e:gradD} yields:
\begin{subequations}
\begin{gather}
\partialderiv{\hat{D}}{\mu_j}=\frac1M\sum^M_{m=1}
{L(\widetilde{\-G}_k(\-y^{(m)})) \, l_k(\-y^{(m)}) \, (2y^{(m)}_j-2\mu_j)}=0\,,\\
\partialderiv{\hat{D}}{\sigma_j}= \frac1M\sum^M_{m=1}
L(\widetilde{\-G}_k(\-y^{(m)}))\, l_k(\-y^{(m)})
\left(\frac{(y^{(m)}_j-\mu_j)^2}{\sigma^3_j}-\frac1{\sigma_j}\right)=0\,,
\end{gather}
\end{subequations}
for $j=1 \ldots n_y$, the solution of which can readily be found as:
\begin{subequations}
\begin{gather}
\mu_j=\frac{\sum^M_{m=1} L(\widetilde{\-G}_k(\-y^{(m)}))l_k(\-y^{(m)})y^{(m)}_j}
{\sum^M_{m=1}L(\widetilde{\-G}_k(\-y^{(m)}))l_k(\-y^{(m)})}\,,\\
{\sigma_j}= 
\sqrt{\frac{\sum^M_{m=1} L(\widetilde{\-G}_k(\-y^{(m)}))l_k(\-y^{(m)})(y^{(m)}_j-\mu_j)^2}
{\sum^M_{m=1}L(\widetilde{\-G}_k(\-y^{(m)}))l_k(\-y^{(m)})}}\,.
\end{gather}
\end{subequations}
This is just an example, of course. We emphasize that our method does
not require any specific type of biasing distribution, and that one
can freely choose the family $\mathcal{P}_V$ that is believed to
contain the best approximations of the posterior distribution.

\subsection{Convergence analysis}
\label{s:convergence}
By design, the tempering parameter $\lambda$ in
Algorithm~\ref{alg:adap} reaches 1 in a finite number of steps, and as
a result the algorithm converges to the original optimization problem.
Thus we only need to analyze the convergence of each step. Without
causing any ambiguity, we will drop the step index $k$ throughout this
subsection.

First, we set up some notation. Let $\widetilde{\-G}_N(\-y)$ be the
$N$-th order polynomial chaos approximation of $\-G(\-y)$, based on a
biasing distribution $q(\-y)$. Also, let
\begin{equation}
D_N(\-v) := \int {L(\widetilde{\-G}_N(\-y)) \ln p(\-y;\-v) \, l(\-y)
  \, q(\-y)} \, d\-y\,,\\
\end{equation}
and
\begin{equation}
\hat{D}_N(\-v)  := \frac1M\sum_{m=1}^M{L \left ( \widetilde{\-G}_N 
    (\-y^{(m)}) \right )
  \ln p(\-y^{(m)};\-v) \, l(\-y^{(m)}}) \, ,
\label{e:D2is}
\end{equation}
where $l(\-y) := \pi(\-y)/q(\-y)$, and $\-y^{(m)}$ are sampled
independently from $q(\-y)$. Note that, for any $\-v$, $\hat{D}_N$ is
a random variable but $D_N$ is a deterministic quantity. We make the
following assumptions:

\begin{assumption}\label{a:1}
\begin{enumerate}
\item[(a)] The biasing distribution $q(\-y)$ satisfies
\begin{equation}
 \|\ln(p(\-y;\-v))l(\-y)\|_{L^2_q}<\infty,
\label{e:fm}
\end{equation}
 where $\|\cdot\|_{L^2_q}$ is the $L^2$-norm with weight $q(\-y)$;
\item[(b)] The likelihood function $L(\-g)$ is bounded and 
  uniformly continuous (with respect to $\-g$) on $\{ \-g=\-G(\-y):
  q(\-y)>0\}\cup\{ \-g=\widetilde{\-G}_N(\-y): q(\-y)>0, \ N \geq 1  \}$.
\end{enumerate}
\end{assumption}

Now we give a lemma that will be used to prove our convergence results:
\begin{lemma}
Suppose that Assumption~\ref{a:1}(a) holds.
If  \[\lim_{N\rightarrow\infty}\|\widetilde{\-G}_N(\-y)-\-G(\-y)\|_{L^2_q}=0,\] 
then 
\[
\lim_{N\rightarrow\infty}\|L(\widetilde{\-G}_N(\-y))- L(\-G(\-y))\|_{L^2_q}=0.
\]
\label{lm:cmt}
\end{lemma}
\begin{proof}
See Appendix.
\end{proof}\\

Our main convergence result is formalized in the proposition below.
\begin{proposition}\label{prop:1}
Suppose that Assumption \ref{a:1} holds. Then we have
 \begin{equation}
 \lim_{M,\,N\rightarrow\infty}\| \hat{D}_N(\-v) - D(\-v)\|_{L^2_q}=0.
\label{e:conv}
 \end{equation}
\end{proposition}

\begin{proof}
The variance of  estimator~\eqref{e:D2is} is
\begin{eqnarray}
\| \hat{D}_N(\-v) - D_N(\-v)\|^2_{L^2_q}
&=&\frac1M \left ( \|L(\widetilde{\-G}_N(\-y)) \ln p(\-y;\-v) \,
  l(\-y)\|^2_{L^2_q}-D_N^2 \right )\notag\\
&\leq&\frac1M \left ( C \| \ln p(\-y;\-v) \, l(\-y)\|^2_{L^2_q}-D_N^2
\right ),
\end{eqnarray}
where $C > 0$ is some constant, and it
follows immediately that
 \begin{equation}
 \lim_{M\rightarrow\infty}\| \hat{D}_N(\-v) - D_N(\-v)\|_{L^2_q}=0.
\label{e:conv1}
 \end{equation}
We then look at the truncation error between $D_N(\-v)$ and $D(\-v)$:
\begin{eqnarray}
|D_N(\-v)-D(\-v)| &=&
\left | \int
  {(L(\widetilde{\-G}_N(\-y))-L(\-G(\-y)))\ln p(\-y;\-v) \, l(\-y)q(\-y)}\,d\-y
  \right |\notag\\
&\leq& \left
  \|(L(\widetilde{\-G}_N(\-y))-L(\-G(\-y)))\ln p(\-y;\-v) \, l(\-y)q(\-y)
  \right \|_{L^1_q}\notag\\
%
%
&\leq&  \|{L(\widetilde{\-G}_N(\-y))}-{L(\-G(\-y))} \|_{L^2_q}\,
\left \|\ln p(\-y;\-v) \, l(\-y) \right \|_{L^2_q},
%
\end{eqnarray}
by H\"{o}lder's inequality. Using Lemma~\ref{lm:cmt} and \eqref{e:fm},
one obtains
\begin{equation}
\lim_{N\rightarrow\infty} |D_N(\-v) -D(\-v)|=0, \label{e:cov2}
\end{equation}
which together with \eqref{e:conv1} implies that $\hat{D}_N(\-v)\rightarrow D(\-v)$ in $L^2_q$,
as $M,\,N\rightarrow\infty$.
\end{proof}

We note that the analysis above is not limited to total-order
polynomial chaos expansions; it is applicable to other polynomial
chaos truncations and indeed to other approximation schemes. What is
required are the conditions of Lemma~\ref{lm:cmt}, where $N$ is any
parameter that indexes the accuracy of the forward model
approximation, such that
$\lim_{N\rightarrow\infty}\|\widetilde{\-G}_N(\-y)-\-G(\-y)\|_{L^2_q}=0$.

\subsection{Independence sampler MCMC}
\label{s:mcmc-indep}

In addition to providing a surrogate model focused on the posterior
distribution, the adaptive algorithm also makes it possible to employ
a Metropolis-Hastings independence sampler, i.e., an MCMC scheme where
the proposal distribution is independent of the present
state~\cite{tierney1994markov} of the Markov chain. When the proposal
distribution is close to the posterior, the independence sampler can
be much more efficient than a standard random walk Metropolis-Hastings
scheme~\cite{gilks1996markov,meyn1993markov,PetraGhattas}, in that it
enables larger ``jumps'' across the parameter space. This suggests
that the final biasing distribution $p(\-y;\-v)$ found by our method
can be a good proposal distribution for use in MCMC simulation.

Let the final biasing distribution obtained by
Algorithm~\ref{alg:adap} be denoted by $p(\-y;\-v_\infty)$, and let
the corresponding surrogate model be $\widetilde{\-G}_\infty(\-y)$. In
each MCMC iteration, the independence sampler updates the current
state $\-y_t$ of the Markov chain via the following steps:
\begin{enumerate}
\item Propose a candidate state $\-y^\prime$ by drawing a sample
  from $p(\-y;\-v_\infty)$.
\item Compute the Metropolis acceptance ratio: 
\begin{equation}
\alpha = \frac{L(\widetilde{\-G}_\infty(\-y^\prime)) \, \pi(\-y^\prime)}
         {L(\widetilde{\-G}_\infty(\-y_t)) \, \pi(\-y_t)} 
     \frac{p(\-y_t;\-v_\infty)}{p(\-y^\prime;\-v_\infty)} \, .
\end{equation}
\item Put $r = \max(\alpha, 1)$. Draw a number $r^\prime \sim U(0,1)$
  and set the next state of the chain to:
\begin{eqnarray*}
\-y_{t+1} = \left \{ \begin{array}{cc}
               \-y^\prime, \ & r^\prime \leq r; \\
                \-y_t, \ & r^\prime > r. 
                \end{array}
               \right .
\end{eqnarray*}
\end{enumerate}

Given the final surrogate $\widetilde{\-G}_\infty(\-y)$, one could
also use a standard random-walk MCMC sampler, or any other valid MCMC
algorithm, to explore the posterior induced by this forward model
approximation. A comparison of the Metropolis independence sampler
with an adaptive random-walk MCMC approach will be provided in
Section~\ref{s:ihcp}.

%% file: example.tex
\section{Numerical examples}\label{s:example}
In this section we present two numerical examples to explore the
efficiency and accuracy of the adaptive surrogate construction
method. The first example is deliberately chosen to be low-dimensional
for illustration purposes. The second is a classic time-dependent
inverse heat conduction problem.

\subsection{Source inversion}
First we will apply our method to the contaminant source inversion
problem studied in \cite{Marzouk2007}, which uses a limited and noisy
set of observations to infer the location of a contaminant source.
Specifically, we consider a dimensionless diffusion equation on a
two-dimensional spatial domain:
\begin{equation}
 \partialderiv{u}{t} = \nabla^2 u + s(\-x,t), \qquad \-x \in D := [0,1]^2,
\label{e:diff}
\end{equation}
with source term $s(\-x,t)$. The field $u(\-x, t)$ represents the
concentration of a contaminant. The source term describes the release
of the contaminant at spatial location $\-x_{\mathrm{src}} :=
(x_1, x_2)$ over the time interval $[0,
\tau]$:
\begin{eqnarray}
s(\-x, t) = \left \{ \begin{array}{cc}
                \frac{s}{2\pi h^2} \exp \left(
                  -{|\-x_{\mathrm{src}} - \-x|^2} / {2 h^{2}} \right), \ & 0
                \leq t \leq \tau \\
                0, \ & t > \tau 
                \end{array}
               \right .
\label{eq:src}
\end{eqnarray}
Here we suppose that the source strength $s$ is known and equal to
2.0, the source width $h$ is known and equal to 0.05, and the source
location $\-x_\mathrm{src}$ is the parameter of interest.  The
contaminant is transported by diffusion but cannot leave the domain;
we thus impose homogeneous Neumann boundary conditions:
$$
\nabla u \cdot \mathbf{n} = 0 \ \ \mathrm{on} \ \partial D .
$$
At the initial time, the contaminant concentration is zero everywhere:
$$
u\left (\-x, 0 \right ) = 0.
$$
The diffusivity is spatially uniform; with a suitable scaling of
space/time, we can always take its value to be unity.

\begin{figure}[!htb]
\centering
\subfigure[Prior-based surrogate.]
{\includegraphics[width=.75\textwidth]{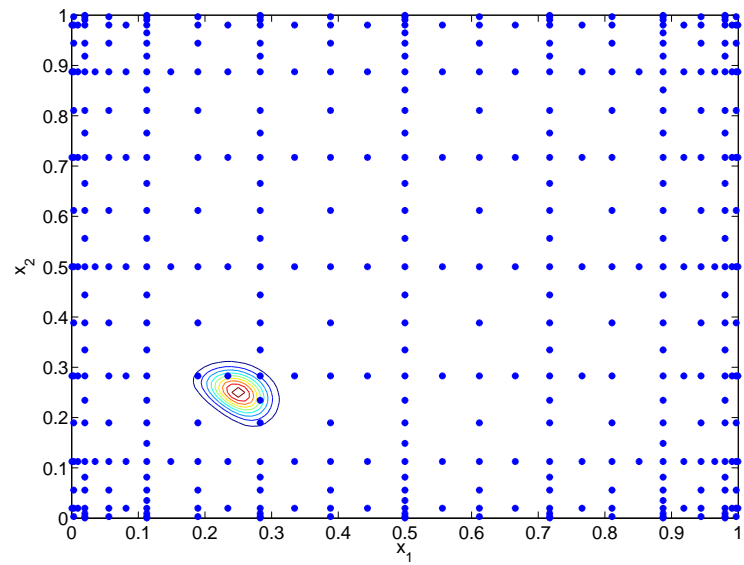}}
\subfigure[Adaptive surrogate.]
{\includegraphics[width=.75\textwidth]{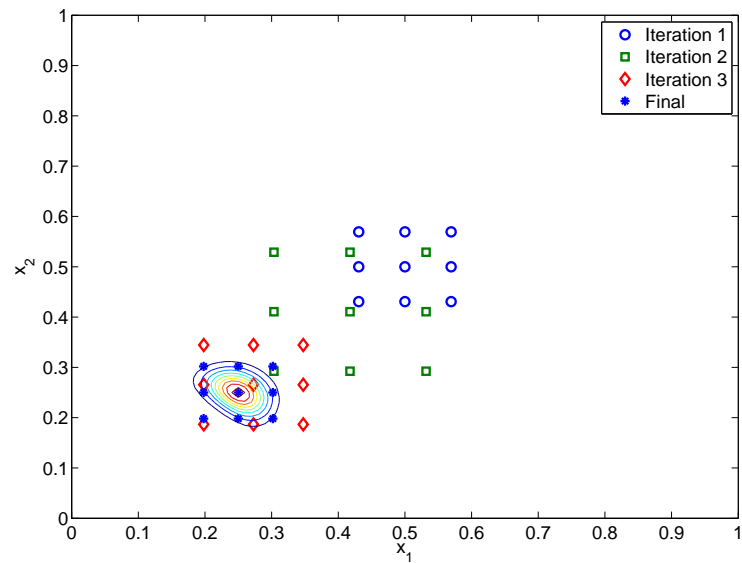}}
\caption{Source inversion problem: model evaluation points used to
  construct the prior-based and adaptive surrogates. In the bottom
  figure (adaptive method), the circles, squares, and diamonds are the
  points evaluated in the first, second, and third iterations,
  respectively. Contours of the posterior probability density are
  superimposed on each figure. }
\label{f:pts_eval}
\end{figure}

\begin{figure}[!htb]
  \centering 
  \subfigure[``True'' posterior density (solid line), compared with
  the posterior density obtained via the prior-based surrogate (dotted
  line).]
  {\includegraphics[width=.75\textwidth]{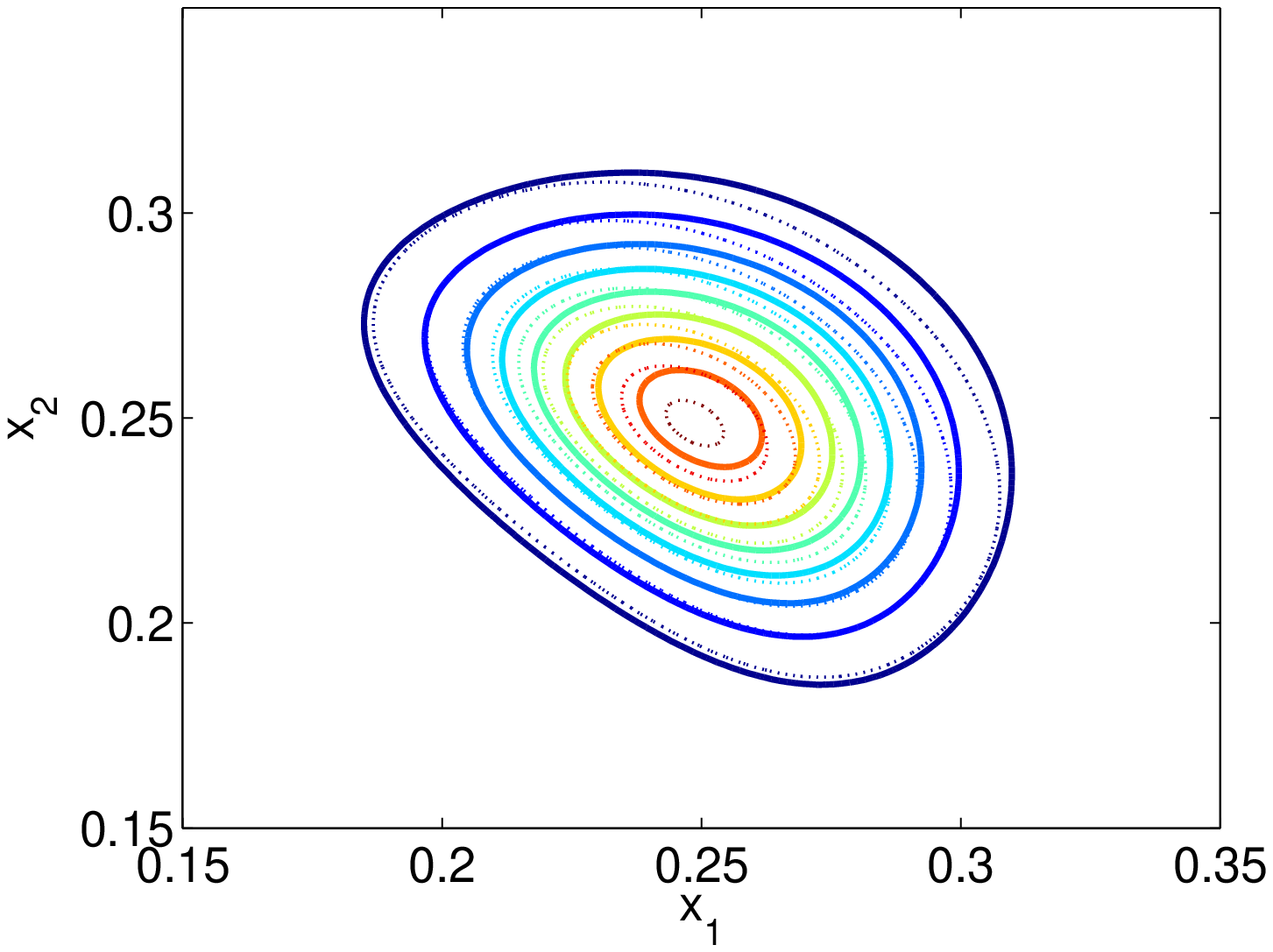}}
  \subfigure[``True'' posterior density (solid line), compared with
  the posterior density obtained via the adaptively-constructed
  surrogate (dashed line).]
  {\includegraphics[width=.75\textwidth]{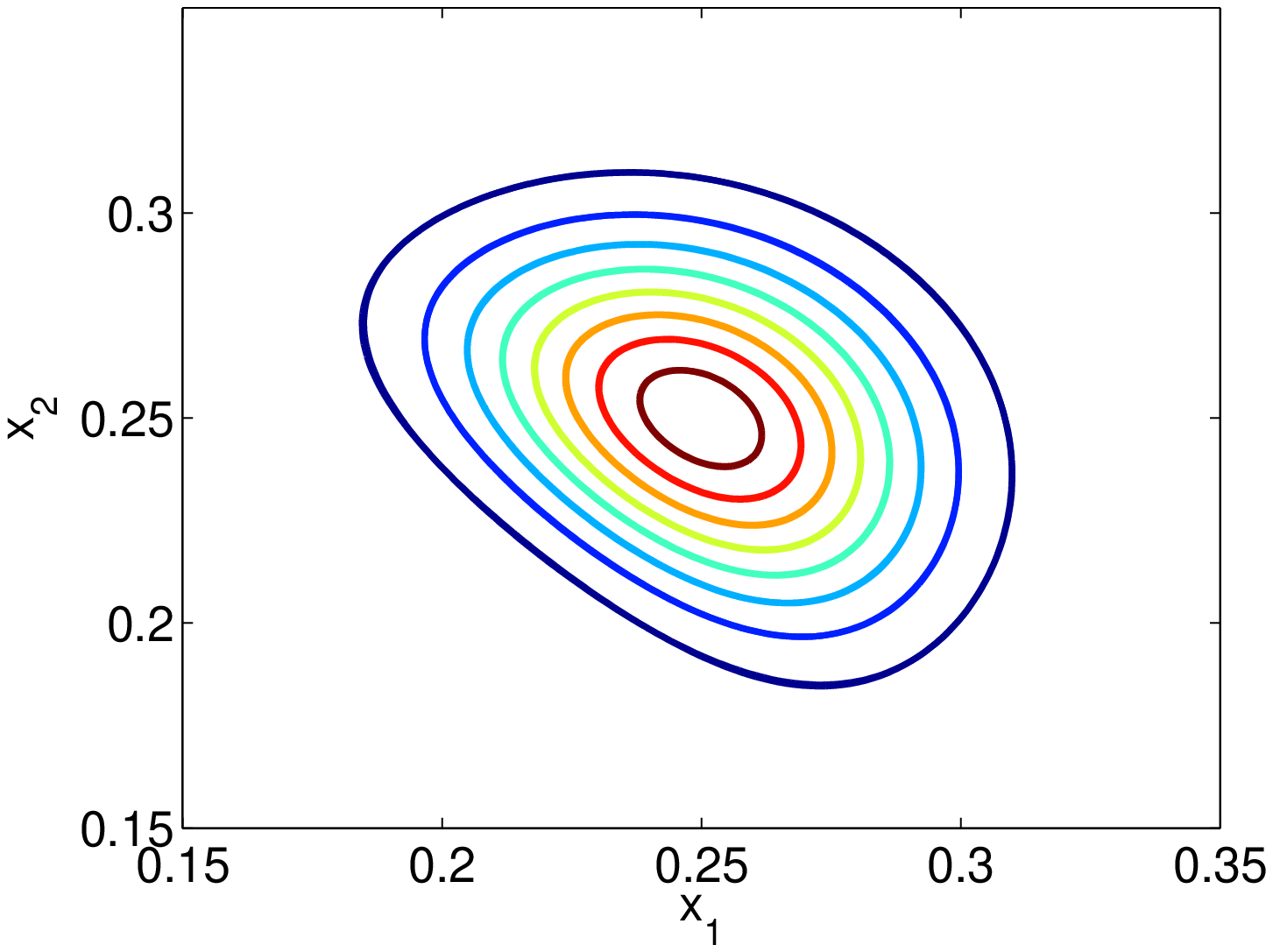}}
\caption{Source inversion problem: posterior density of
  $\mathbf{x}_\mathrm{src}$ obtained with the three different
  approaches. In the bottom figure, the two sets of contours are
  virtually indistinguishable.}
\label{f:post}
\end{figure}

Sensors that measure local concentration values are placed on a
uniform $3 \times 3$ grid covering $D$, and sensor readings are
provided at two successive times, $t=0.1$ and $t=0.2$, resulting in a
total of 18 measurements.  The forward model thus maps the source
position to the values of the field $u(\-x,t)$ at the prescribed
measurement locations and times, while the inverse problem consists of
inferring the source position from noisy measurements.  We write the
forward model as $\-d=\-G(\-x_\mathrm{src})+\bolde$, where $\-d$ is
the vector collecting all the measurements and $\bolde$ is the
measurement error.  Each component of $\bolde$ is assumed to be an
independent zero-mean Gaussian random variable: $\epsilon_i \sim
N(0,\sigma^2)$ with $\sigma=0.1$ for $i=1 \ldots 18$.  In this example
we generate simulated data $\-d$ by solving the forward model with
$\-x_\mathrm{src}=(0.25,\,0.25)$ and adding noise. To complete the
Bayesian setup, we take the prior to be a uniform distribution over
$D$; that is, $x_j \sim U(0,1)$ for $j=1,\,2$.

First we characterize the posterior distribution using the proposed
adaptive method.  We fix the polynomial order of the surrogates
$\widetilde{\-G}_k( x_1, x_2 )$ to $N=4$ and take the biasing
distribution to be an uncorrelated Gaussian:
\begin{equation}
p(x_1,x_2; \-v) = \frac1{2\pi\sigma_1\sigma_2}\exp\left(-\frac{(x_1-\mu_1)^2}{2\sigma_1^2}-\frac{(x_2-\mu_2)^2}{2\sigma_2^2}\right)\,,
\end{equation}
where the means $\mu_1$, $\mu_2$ and standard deviations $\sigma_1$,
$\sigma_2$ comprise the reference parameters $\-v$.
The initial biasing distribution is centered at the prior mean and has
small variance; that is, we choose $\mu_1= \mu_2=0.5$ and $\sigma_1=
\sigma_2=0.05$. We set  $\rho = 0.05$, $\gamma = 10^{-3}$, and 
put $\delta = (\lambda_0-1)/10$ (namely, we require $\lambda$ to reach $1$ in at most 10 iterations). 
In each iteration, a $3\times3$ tensor product
Gaussian-Hermite quadrature rule is used to construct a Hermite polynomial
chaos surrogate, resulting in nine true model evaluations; $M=5\times
10^4$ surrogate samples are then employed to estimate the reference
parameters. 
It takes four iterations for the algorithm to converge, and its main
computational cost thus consists of evaluating the true model 36 times.

As a comparison, we also construct a polynomial surrogate with respect
to the uniform prior distribution. Here we use a surrogate composed of
total order $N=9$ Legendre polynomials, and compute the polynomial
coefficients with a level-6 sparse grid based on Clenshaw-Curtis
quadrature, resulting in 417 true model evaluations---about 12 times
as many as the adaptive method. These values were chosen so that the
prior-based and adaptive polynomial surrogates have comparable (though
not exactly equal) accuracy.

In Figure~\ref{f:pts_eval}, we show the points in parameter space at
which the true model was evaluated in order to construct the two types
of surrogates. Contours of the posterior density are superimposed on
the points. It is apparent that in the prior-based method, although
417 model evaluation points are used, only five of them actually fall
in the region of significant posterior probability. With the adaptive
method, on the other hand, 8 of the 36 model evaluations occur in the
important region of the posterior distribution. Figure~\ref{f:post}
shows the posterior probability densities resulting from both types of
surrogates. Since the problem is two-dimensional, these contours were
obtained simply by evaluating the posterior density on a fine grid,
thus removing any potential MCMC sampling error from the problem.
Also shown in the figure is the posterior density obtained with direct
evaluations of the true forward model, i.e., the ``true'' posterior.
While both surrogates provide reasonably accurate posterior
approximations, the adaptive method is clearly better; its posterior
is essentially identical to the true posterior. Moreover, the adaptive
method requires an order of magnitude fewer model evaluations than the
prior-based surrogate.

\subsection{Inverse heat conduction}
\label{s:ihcp}
Estimating temperature or heat flux on an inaccessible boundary from
the temperature history measured inside a solid gives rise to an
inverse heat conduction~(IHC) problem. These problems have been
studied for several decades due to their significance in a variety of
scientific and engineering
applications~\cite{beck1985inverse,Ozisik2000,Wang200515,macayeal1991paleothermometry}.
An IHC problem becomes nonlinear if the thermal properties are
temperature-dependent~\cite{beck1970nonlinear,carasso1982determining};
this feature renders inversion significantly more difficult than in the linear
case.  In this example we consider a one-dimensional heat conduction
equation:
\begin{equation}
\partialderiv{u}t=\partialderiv{}x\left[c(u)\partialderiv{u}x\right],
\label{e:heat}
\end{equation}
where $x$ and $t$ are the spatial and temporal variables, $u(x,t)$ is
the temperature, and 
$$
c(u) := \frac{1}{1+u^2}
$$ 
is the temperature-dependent thermal conductivity, all in dimensionless form. 
The equation is subject to initial condition $u(x,0)=u_0(x)$ and 
Neumann boundary conditions:
\begin{subequations}
\begin{gather}
\partialderiv{}xu(0,t)=q(t),\\
\partialderiv{}xu(L,t) = 0,
\end{gather}
\end{subequations}
where $L$ is the length of the medium. In other words, one end ($x=L$)
of the domain is insulated and the other ($x=0$) is subject to heat
flux $q(t)$. Now suppose that we place a temperature sensor at
$x=x_s$. The goal of the IHC problem is to infer the heat flux $q(t)$
for $t\in[0,\,T]$ from the temperature history measured at the sensor
over the same time interval. A schematic of this problem is shown in
Figure~\ref{f:heat}.

\begin{figure}[htb]
\centering
\includegraphics[width=.75\textwidth]{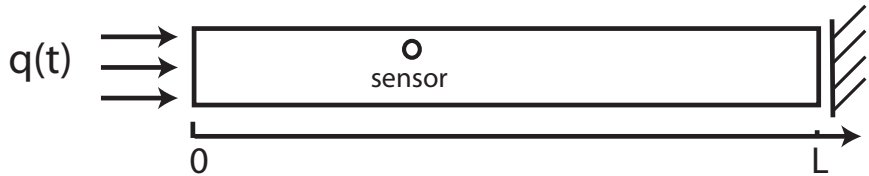}
\caption{Schematic of the one-dimensional heat conduction problem.}
\label{f:heat}
\end{figure}

For the present simulations, we put $L=1$ and $T=1$ and let the initial condition be $u_0(x)=0$.
We parameterize the flux signal with a Fourier series:
\begin{equation}
q(t)= a_0+\sum_{j=1}^{N_f} \left ( a_j\cos(2j\pi t/T)+b_j\sin(2j\pi
  t/T) \right ),\qquad 0\leq t\leq T,
\label{e:flux}
\end{equation}
where $a_j$ and $b_j$ are the coefficients of the cosine and sine
components, respectively, and $N_f$ is the total number of Fourier
modes.  In the following tests, we will fix $N_f=4$; the inverse
problem is thus 9-dimensional. The sensor is placed at $x_s=0.4$,
the temperature is measured at 50 regularly spaced times over the time
interval $[0,T]$, and the error in each measurement is assumed to be
an independent zero-mean Gaussian random variable with variance
$\sigma^2 = 10^{-2}$.

To generate data for inversion, the ``true'' flux is chosen to be 
\begin{equation}
q_\mathrm{true}(t)= \sum_{j=1}^{4} (1.5\cos(2j\pi t)+1.5\sin(2j\pi t))\,,
\end{equation}
i.e., $a_0=0$, $a_j=1.5$, and $b_j=1.5$ for $j=1 \ldots 4$.
Figure~\ref{f:solution} shows the entire solution of \eqref{e:heat}
with the prescribed ``true'' flux, with the sensor location indicated
by the dashed line. The inverse problem becomes more ill-posed as the
sensor moves to the right, away from the boundary where $q$ is
imposed; information about the time-varying input flux is
progressively destroyed by the nonlinear
diffusion. Figure~\ref{f:temps} makes this fact more explicit, by
showing the temperature history at $x=0$ (i.e., the boundary subject
to the heat flux) and at $x=0.4$ (where the sensor is placed). The
data for inversion are generated by perturbing the latter profile with
the i.i.d.\ observational noise. A finer numerical discretization of
\eqref{e:heat} is used to generate these data than is used in the
inference process.
To complete the Bayesian setup, we endow the Fourier coefficients with
independent Gaussian priors that have mean zero and variance $2$.

\begin{figure}[!htb]
\centering
\subfigure[Solution of the nonlinear diffusion equation \eqref{e:heat} as a function
of space and time. The dashed line indicates the location of the
sensor.]
{\includegraphics[width=.75\textwidth]{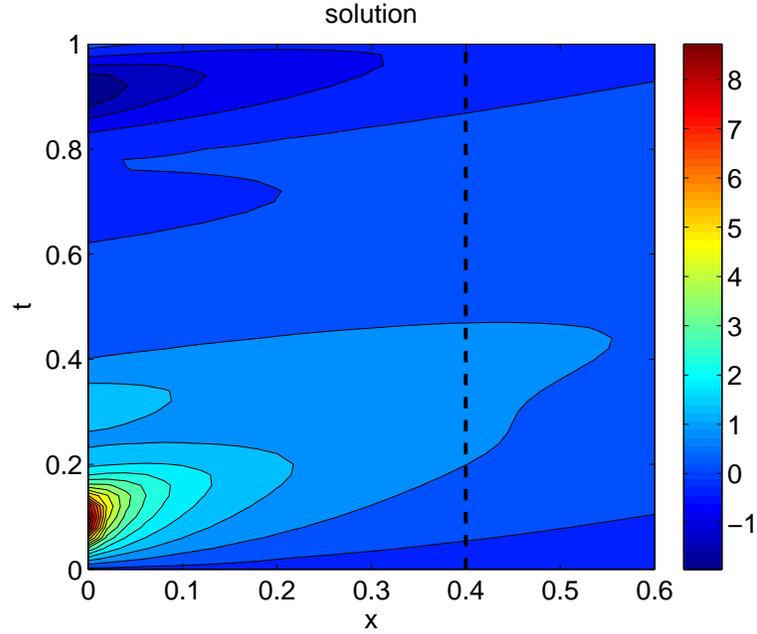}
\label{f:solution}}
\centering
\subfigure[Temperature history at $x=0$ (the boundary where the
time-dependent heat flux $q(t)$ is imposed, solid line); and at
$x=0.4$ (location of the sensor, dashed line).]
{\includegraphics[width=.7\textwidth]{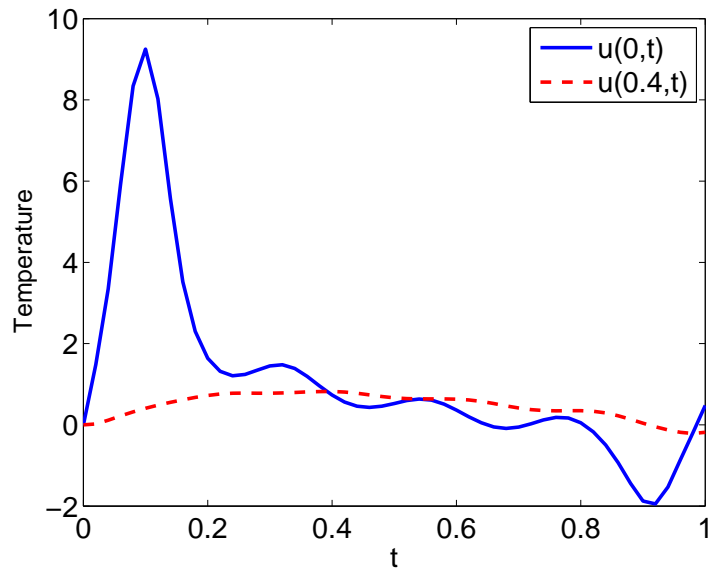}
\label{f:temps}}
\caption{Forward solution of the transient heat conduction problem.}
\end{figure}

\begin{table}
\centering

  \begin{tabular}{c|c|c|c}
    \hline
		& number of & $\dis_\mathrm{KL}(\pi^\ast \| \,
                \pi_\mathrm{PC})$ & $\dis_\mathrm{KL}(\pi^\ast \| \, \pi_\mathrm{PC})$ \\
    surrogate & model evaluations& $\pi(a_1,b_1)$ & $\pi(a_4,b_4)$ \\ \hline
    {prior-based}&  &  \\
     $N=3, S=6$  &11833 &124 &1.89\\ \hline
      \textbf{prior-based}& &  \\
     {$\bm{N=5,\, S=7}$}  &35929 &8.37& 0.383 \\ \hline
   { adaptive} &  &  \\
   {$N=2,  S= 3$} & $2445$ & 0.0044&0.0129\\
    \hline
    \textbf{adaptive}&  &  \\
    {$\bm{N=2,\, S=5}$} & $4459$ & 0.0032&0.0127\\
    \hline
  \end{tabular}
	
	\medskip

	\caption{Cost and performance comparison of prior-based and
          adaptively-constructed surrogates.}
        \label{t:numevals}
\end{table}

We now use the adaptive algorithm to construct a localized polynomial
chaos surrogate, and compare its computational cost and performance to
that of a prior-based surrogate. In both cases, we use Hermite
polynomial chaos to approximate the forward model, and sparse grids
based on tensorization of univariate delayed Kronrod-Patterson
rules~\cite{petras2003smolyak} to compute the polynomial coefficients
(non-intrusively). To test the prior-based method, we use two
different total-order truncations of the PC expansion, one at $N=3$
(with sparse grid level $S=6$) and the other at $N=5$ (with sparse
grid level $S=7$). Sparse grid levels were chosen to ensure relatively
small aliasing error in both cases.

To run the adaptive algorithm, we set  $\rho = 0.05$, $\gamma = 10^{-3}$, and
$\delta = (\lambda_0-1)/20$ (see Algorithm~\ref{alg:adap}) and we use
$M=10^{5}$ samples
for importance sampling at each step. The initial biasing distribution
is centered at the prior mean, with the variance of each component set
to $0.5$.
As described in \eqref{e:normal}, the biasing distributions are chosen
simply to be uncorrelated Gaussians. The optimization procedure takes
14 iterations to converge, and at each iteration, a new surrogate with
$N=2$ and $S=3$ is constructed with respect to the current biasing
distribution. Once the final biasing distribution is obtained, we
construct a corresponding \textit{final} surrogate of the same
polynomial order, $N=2$. Here we typically employ the same sparse grid
level as in the adaptive iterations, but we also report results for a
higher sparse grid level ($S=5$) just to ensure that aliasing errors
are small. The total number of full model evaluations associated with
the adaptive procedure, contrasted with the number used to construct
the prior-based surrogates, is shown in the second column of
Table~\ref{t:numevals}.


With various final surrogates $\widetilde{\-G}$ in hand (both
prior-based and adaptively constructed), we now replace the exact
likelihood function $L(\-G)$ with $L(\widetilde{\-G})$ to generate a
corresponding collection of posterior distributions for comparison. We
use a delayed-rejection adaptive Metropolis (DRAM) MCMC
algorithm~\cite{haario:2006:dra} to draw $5\times10^5$ samples from each
distribution, and discard the first $10^4$ samples as burn-in.  To
examine the results, we cannot visualize the 9-dimensional posteriors
directly; instead we consider several ways of extracting posterior
information from the samples.

First, we focus on the Fourier coefficients
directly. Figure~\ref{fig:coefficients} shows kernel density
estimates~\cite{ihlerkernel} of the marginal posterior densities of
Fourier coefficients of the heat flux $q(t)$. Figure~\ref{f:a1b1}
shows coefficients of the lowest frequency modes, while
Figure~\ref{f:a4b4} shows coefficients of the highest-frequency
modes. In each figure, we show the posterior densities obtained by
evaluation of the exact forward forward model, evaluation of the
adaptive surrogate, and evaluation of the prior-based surrogate. The
latter two surrogates correspond to the second and fourth rows of
Table~\ref{t:numevals} (marked in bold type). Even though
construction of the prior-based surrogate employs more than 6 times as
many model evaluations as the adaptive algorithm, the adaptive
surrogate is far more accurate. This assessment is made quantitative
by evaluating the K-L divergence from the exact posterior distribution
to each surrogate-induced posterior distribution (focusing only on the
two-dimensional marginals). Results are shown in the last two columns
of Table~\ref{t:numevals}. By this measure, the adaptive surrogate is
three orders of magnitude more accurate in the low-frequency modes and
at least an order of magnitude more accurate in the high-frequency
modes. (Here the K-L divergences have also been computed from the
kernel density estimates of the pairwise posterior marginals. Sampling
error in the K-L divergence estimates is limited to the last reported
digit.) 
The difference in accuracy gains between the low- and high-frequency
modes may be due to the fact that the posterior concentrates more
strongly for the low-frequency coefficients, as these are the modes
for which the data/likelihood are most informative. Thus, while the
adaptive surrogate here provides higher accuracy in \textit{all} the
Fourier modes, improvement over the prior surrogate is expected to be
most pronounced in the directions where posterior concentration is
greatest.


\begin{figure}[htb]
\centering
\subfigure[$\pi^\ast(a_1, b_1).$]
{\includegraphics[width=.75\textwidth]{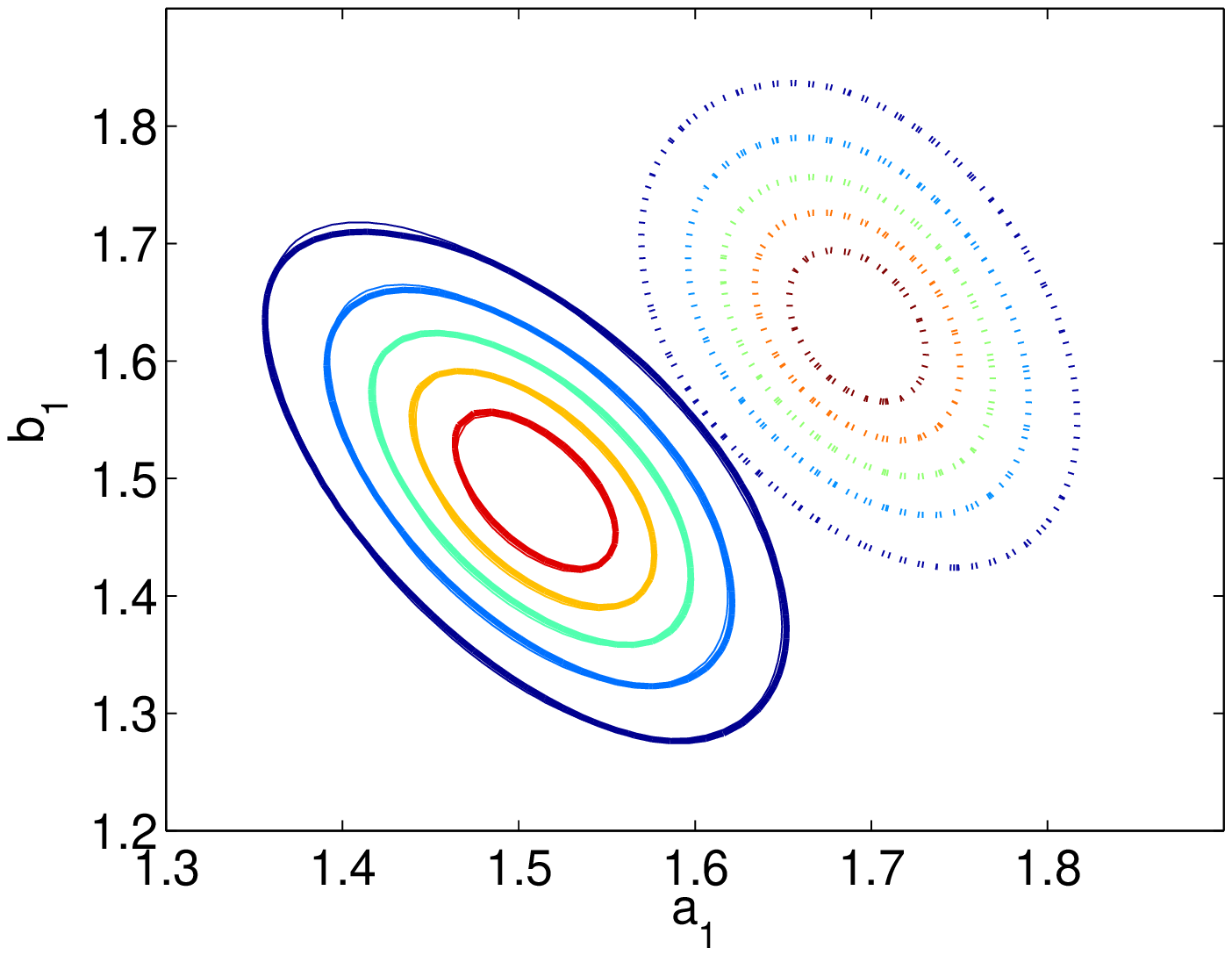}}
\label{f:a1b1}
\centering
\subfigure[$\pi^\ast(a_4, b_4).$ ]
{\includegraphics[width=.75\textwidth]{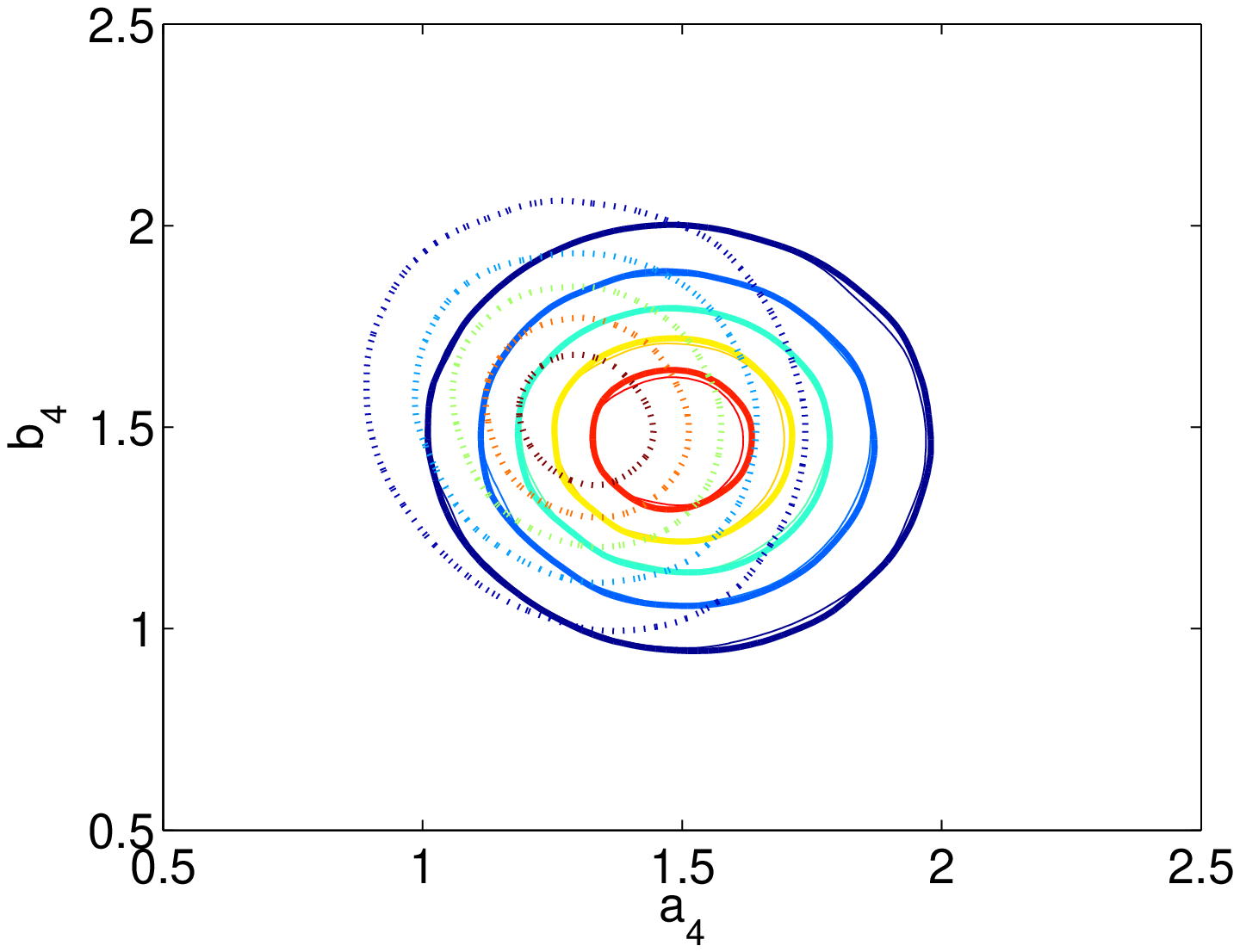}}
\label{f:a4b4}
\caption{Inverse heat conduction problem: marginal posterior densities
  of pairs of Fourier coefficients. Solid lines (thicker) are results
  of the exact forward model; solid lines (thinner) are from the
  adaptive surrogate; dotted lines are from the prior-based
  surrogate.}
\label{fig:coefficients}
\end{figure}


To further assess the performance of the adaptive method, we use
posterior samples of the Fourier coefficients to reconstruct posterior
moments of the heat flux itself, again using the exact forward model,
the adaptively-constructed surrogate, and the prior
surrogate. Figure~\ref{f:meanq} shows the posterior mean,
Figure~\ref{f:meanq} shows the posterior variance, and
Figure~\ref{f:skewq} shows the posterior skewness; these are moments
of the \textit{marginal} posterior distributions of heat flux $q(t)$
at any given time $t$. Again, we show results only for the surrogates
identified in the second and fourth lines of
Table~\ref{t:numevals}. The adaptively-constructed (and posterior-focused)
surrogate clearly outperforms the prior-based surrogate. Note that the
nonzero skewness is a clear indicator of the non-Gaussian character of
the posterior.

Moving from moments of the marginal distributions to correlations
between heat flux values at different times, Figure~\ref{f:cov_apc}
shows the posterior autocovariance of the heat flux computed with the
adaptive surrogate, which also agrees well with the values computed
from the true model.

\begin{figure}[htb]
\centerline{\includegraphics[width=.75\textwidth]{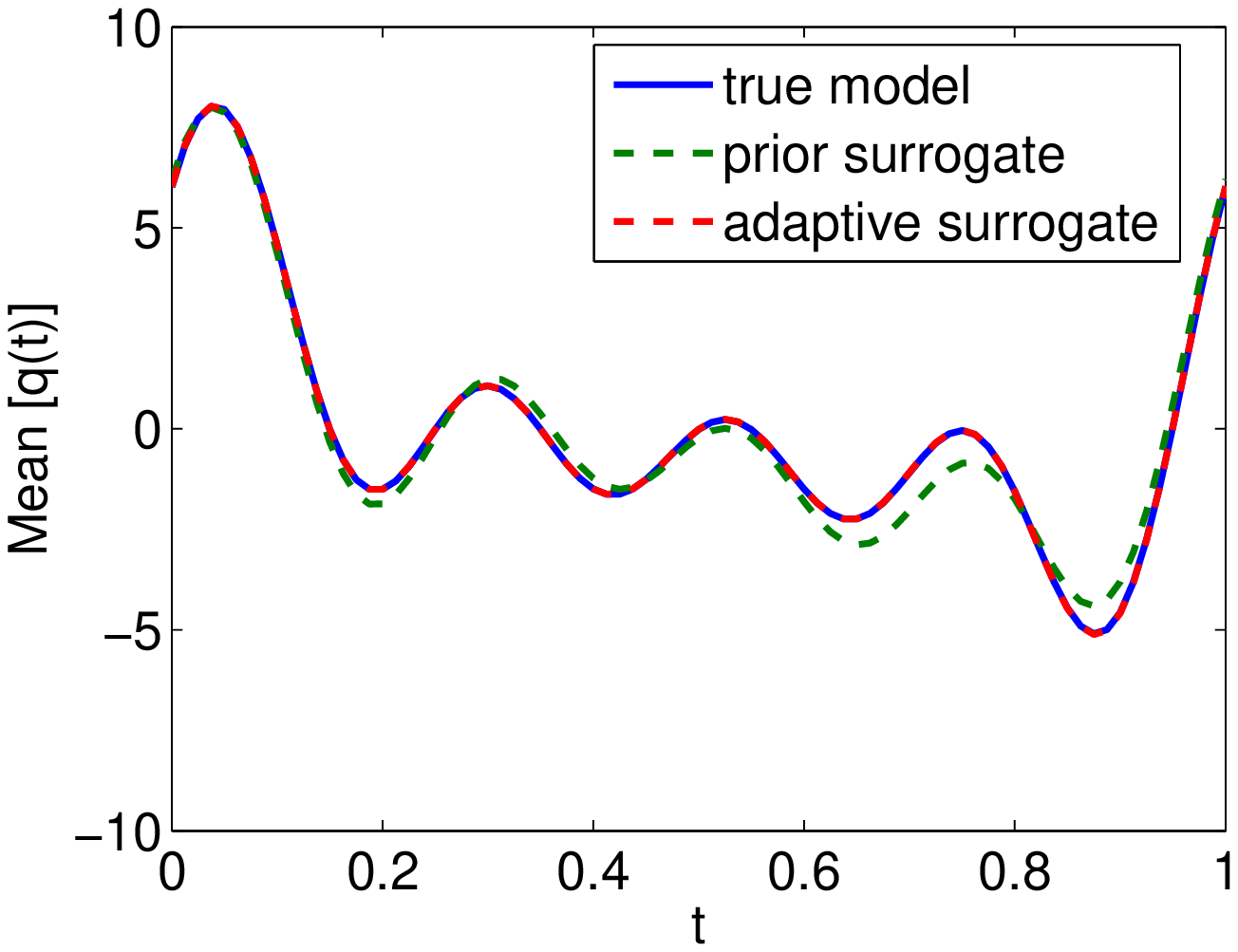}}
\caption{Inverse heat conduction problem: posterior mean of the flux
  $q(t)$ computed with the true model, the prior-based surrogate, and
  the adaptively-constructed surrogate.}
\label{f:meanq}
\end{figure}

\begin{figure}[htb]
\centerline{\includegraphics[width=.75\textwidth]{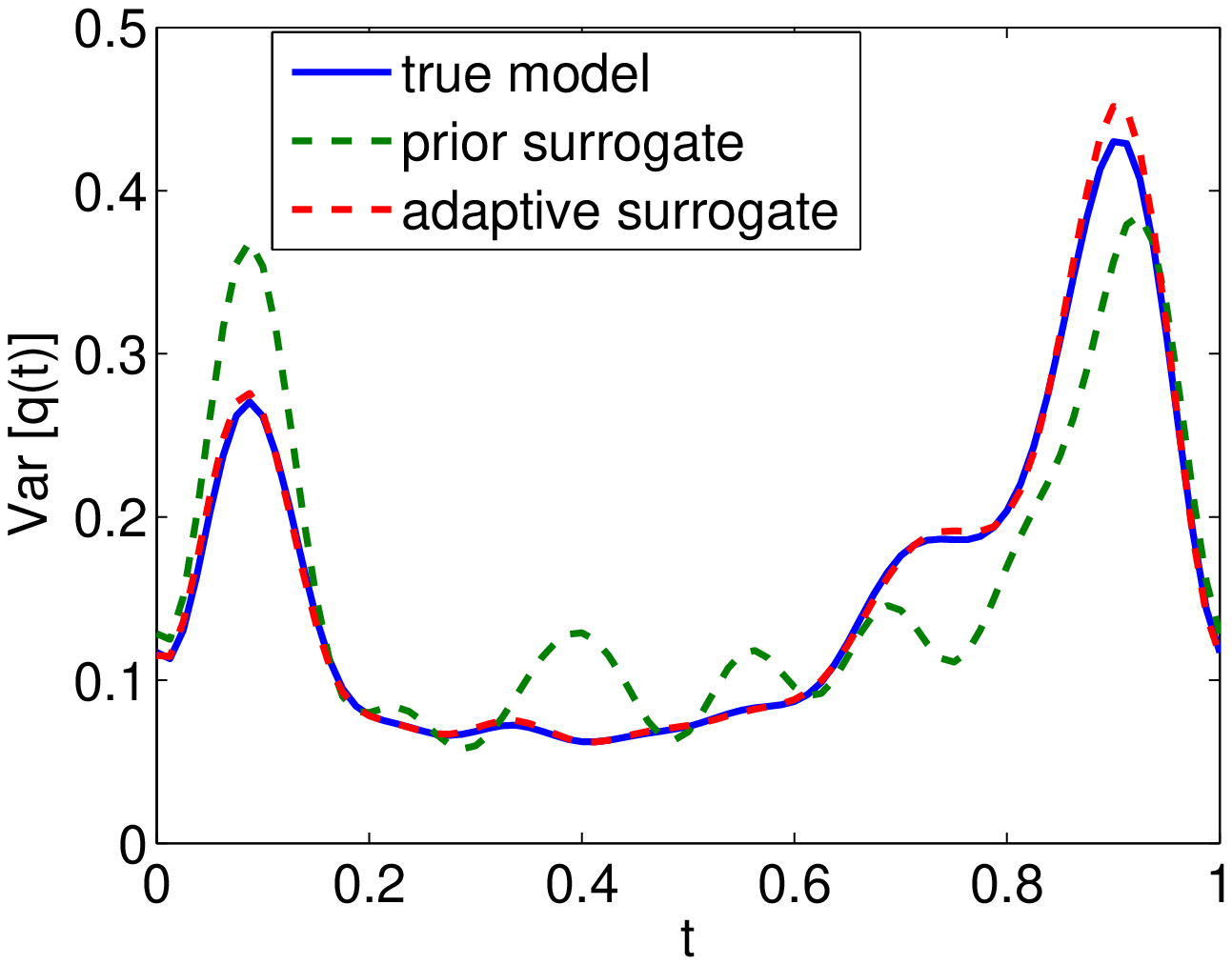}}
\caption{Inverse heat conduction problem: posterior variance of the flux
  $q(t)$ computed with the true model, the prior-based surrogate, and
  the adaptively-constructed surrogate.}
\label{f:varq}
\end{figure}

\begin{figure}[htb]
\centerline{\includegraphics[width=.75\textwidth]{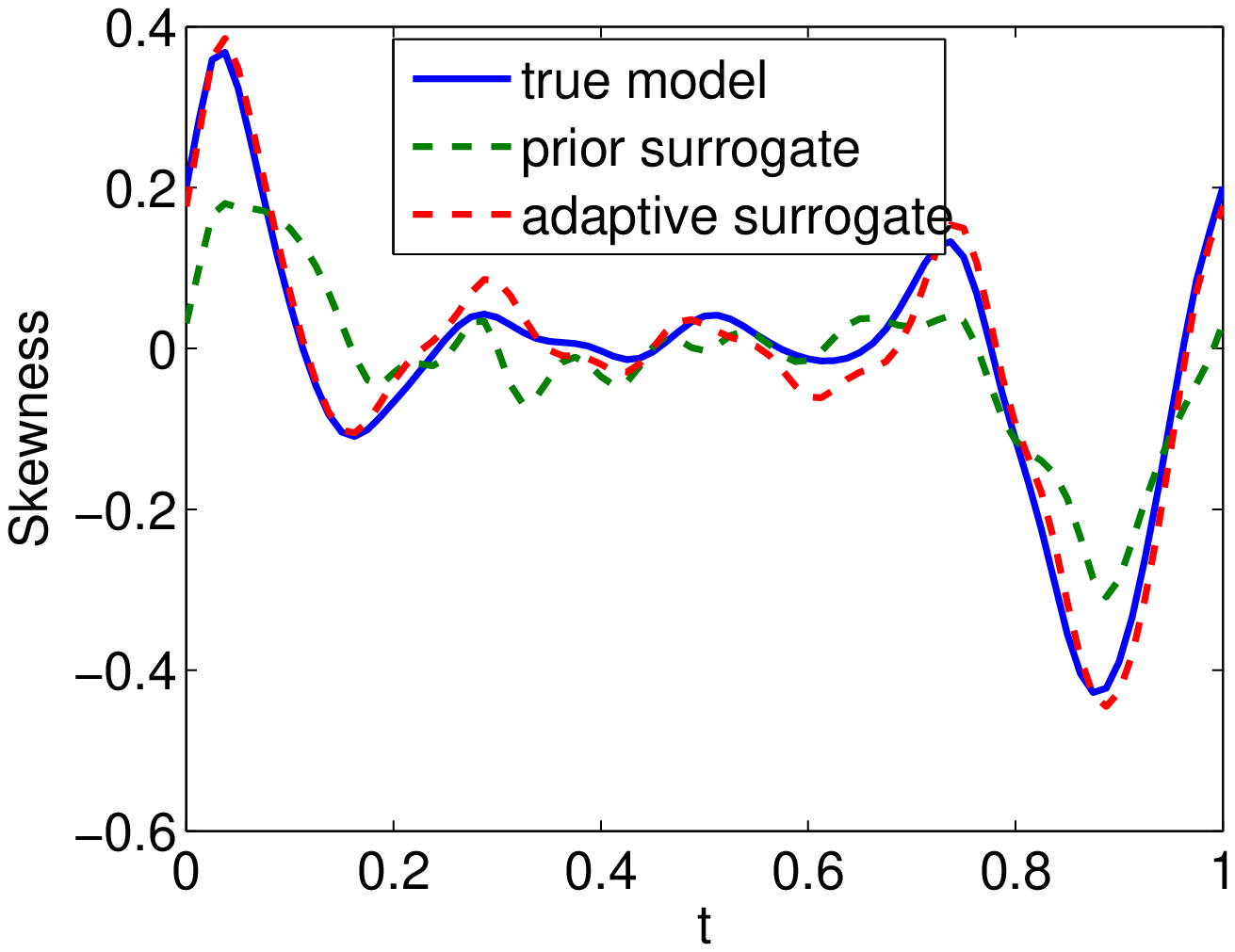}}
\caption{Inverse heat conduction problem: posterior skewness of the flux
  $q(t)$ computed with the true model, the prior-based surrogate, and
  the adaptively-constructed surrogate.}
\label{f:skewq}
\end{figure}


\begin{figure}[htb]
\centerline{\includegraphics[width=.75\textwidth]{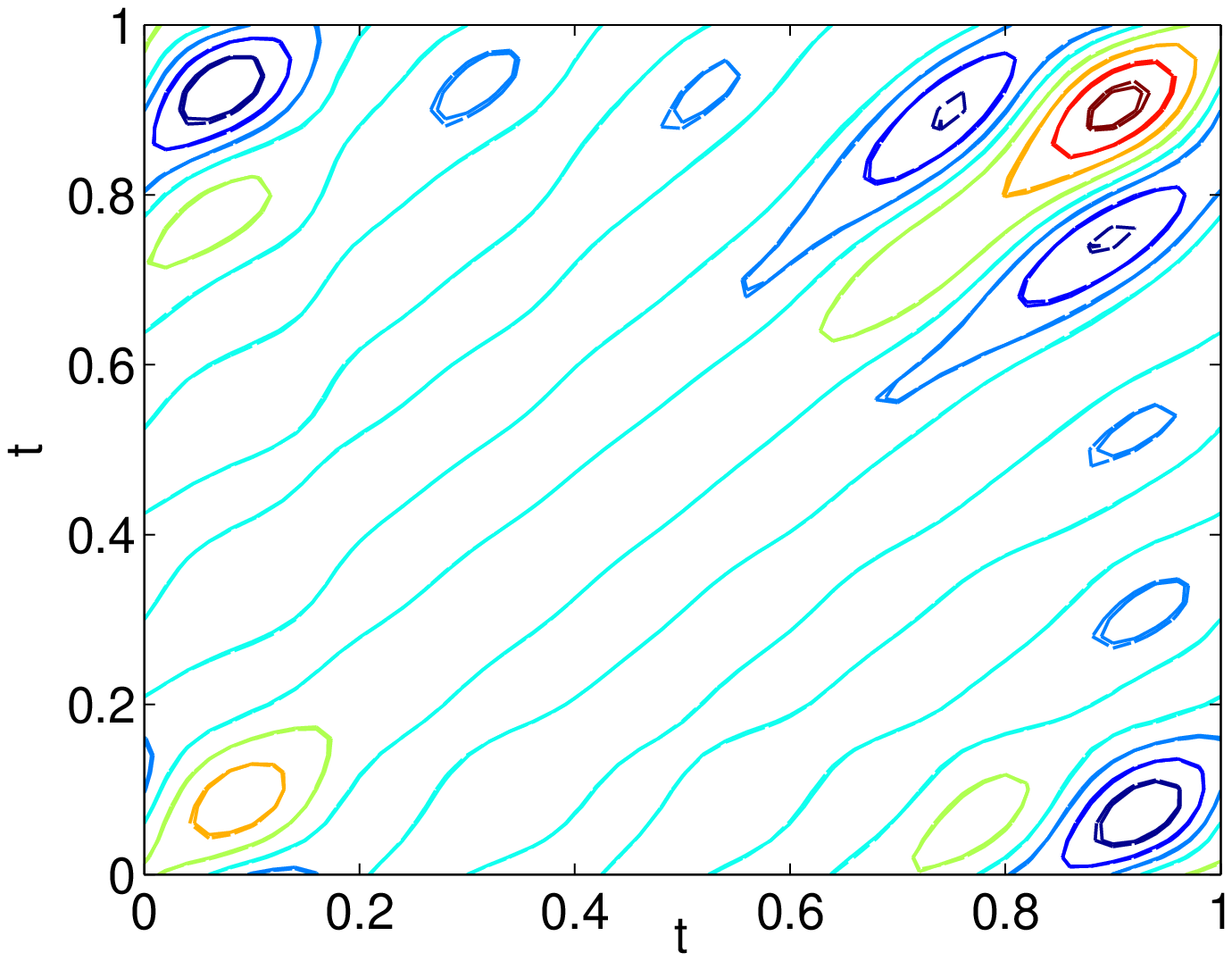}}
  \caption{Inverse heat conduction problem: posterior covariance of the flux $q(t)$. Solid lines are
    computed from the exact-model posterior, while dashed lines are
    computed with the adaptive surrogate.}
\label{f:cov_apc}
\end{figure}


Finally, we evaluate the MCMC independence sampler proposed in
Section~\ref{s:mcmc-indep}, comparing its performance with that of the
adaptive random-walk sampler (DRAM). For the present inverse heat
conduction problem, the mixing of each sampler is essentially
independent of the particular posterior (exact or surrogate-based) to
which it is applied; we therefore report results for the adaptive
surrogate only. Figure~\ref{f:autocorr} plots the empirical
autocorrelation of the MCMC chain as a function of lag, computed from
$5\times10^{5}$ iterations of each sampler. We focus on the MCMC chain
of $a_1$, but the relative autocorrelations of other Fourier coefficients are
similar. Rapid decay of the autocorrelation is indicative of good
mixing; MCMC iterates are less correlated, and the variance of any
MCMC estimate at a given number of iterations is
reduced. Figure~\ref{f:autocorr} shows that autocorrelation decays
considerably faster with the independence sampler than with the
random-walk sampler, which suggests that the final biasing
distribution computed with the adaptive algorithm can indeed be a good
proposal distribution for MCMC.

\begin{figure}[htb]
\centering
\includegraphics[width=.75\textwidth]{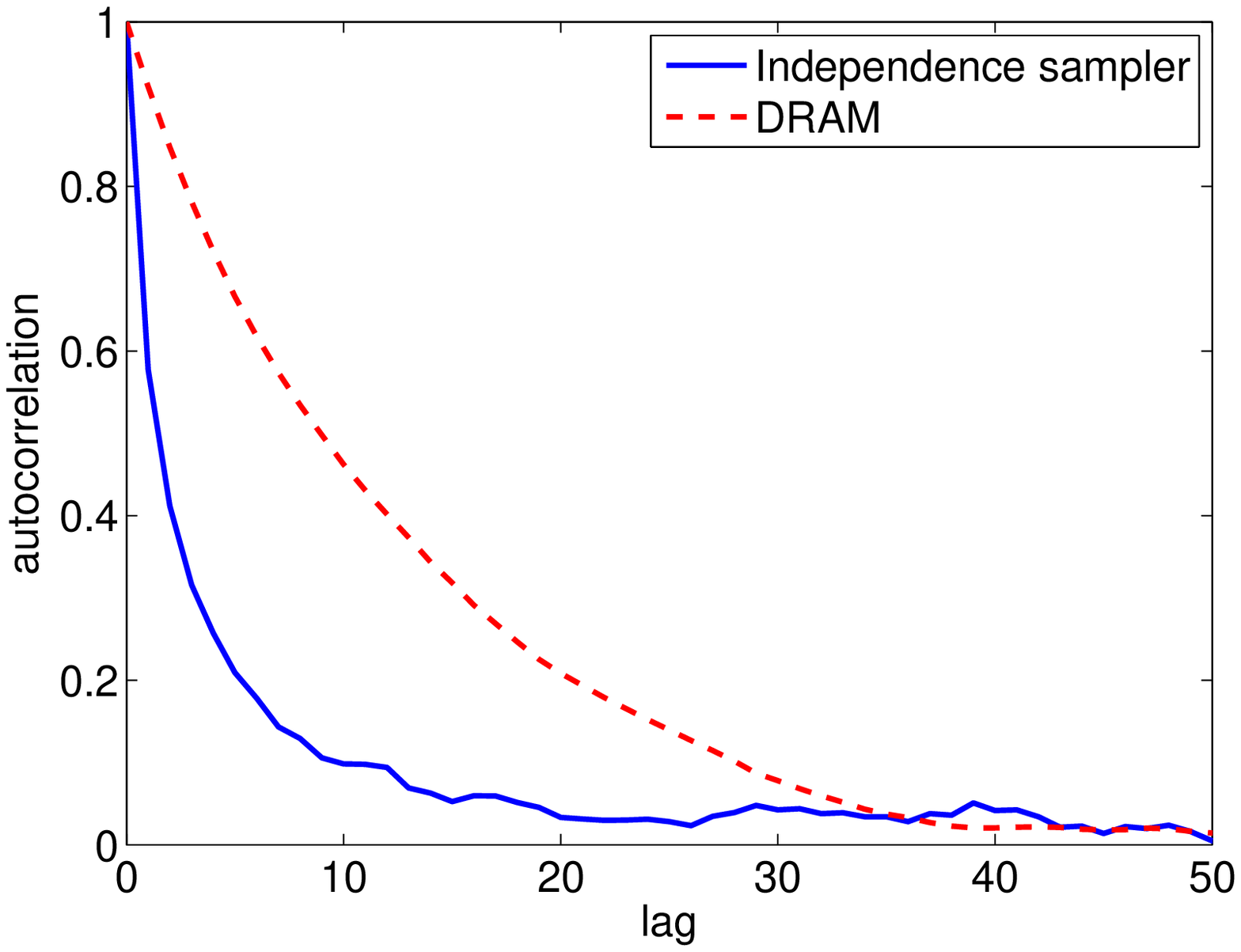}
\caption{Chain autocorrelation of an MCMC independence sampler derived
  from the adaptive algorithm, versus an adaptive random-walk MCMC
  sampler.}
\label{f:autocorr}
\end{figure}

%% file: conclusion.tex
\section{Conclusions}
\label{s:conclusion}

This paper has developed an efficient adaptive approach for
approximating computationally intensive forward models for use in
Bayesian inference. These approximations are used as surrogates for
the exact forward model or parameter-to-observable map, thus making
sampling-based Bayesian solutions to the inverse problem more
computationally tractable.

The present approach is \textit{adaptive} in the sense that it uses
the data and likelihood function to focus the accuracy of the forward
model approximation on regions of high posterior probability. This
focusing is performed via an iterative procedure that relies on
stochastic optimization. In typical inference problems, where the
posterior concentrates on a small fraction of the support of the prior
distribution, the adaptive approach can lead to significant gains in
efficiency and accuracy over previous methods. Numerical
demonstrations on inference problems involving partial differential
equations show order-of-magnitude increases in computational efficiency
(as measured by the number of forward solves) and accuracy (as
measured by posterior moments and information divergence from the
exact posterior) over prior-based surrogates employing comparable
approximation schemes.

The adaptive algorithm generates a finite sequence of biasing
distributions from a chosen parametric family, and accelerates the
identification of these biasing distributions by constructing
approximations of the forward model at each step. The final biasing
distribution in this sequence minimizes Kullback-Leibler divergence
from the true posterior; convergence to this minimizer is assured as
the number of samples (in an internal importance sampling estimator)
goes to infinity and the accuracy of the local surrogate is
increased. As a byproduct of the algorithm, the final biasing
distribution can also serve as a useful proposal distribution for
Markov chain Monte Carlo exploration of the posterior distribution.

Since the adaptive approach relies on concentration of the posterior
relative to the prior, it is best suited for inference problems where
the data are informative in the same relative sense. Yet most
``useful'' inference problems will fall into this category.  The more
difficult it is to construct a globally accurate surrogate (for
instance, as the forward model becomes more nonlinear) and the more
tightly the posterior concentrates, the more beneficial the adaptive
approach may be. Now in an ill-posed inverse problem, the posterior
may concentrate in some directions but not in others; for instance,
data in the inverse heat conduction problem are less informative about
higher frequency variations in the inversion parameters. Yet
significant concentration does occur overall, and it is largest in the
directions where the likelihood function varies most and is thus most
difficult to approximate. This correspondence is precisely to the
advantage of the adaptive method.

We note that the current algorithm does not require access to
derivatives of the forward model. If derivatives were available and
the posterior mode could be found efficiently, then it would be
natural to use a Laplace approximation at the mode to initialize the
adaptive procedure. Also, an expectation-maximization (EM) algorithm
could be an interesting alternative to the adaptive importance sampling
approach used to solve the optimization problem in~\eqref{e:maxD}.
Localized surrogate models could be employed to accelerate
computations within an EM approach, just as they are used in the
present importance sampling algorithm.


Finally, we emphasize that while the current numerical demonstrations
used polynomial chaos approximations and Gaussian biasing
distributions, the algorithm presented here is quite general. Future
work could explore other families of biasing distributions and other
types of forward model approximations---even projection-based model
reduction schemes. Simple parametric biasing distributions could also
be replaced with finite mixtures (e.g., Gaussian mixtures), with a
forward model approximation scheme tied to the components of the
mixture. Another useful extension of the adaptive algorithm could
involve using full model evaluations from previous iterations to
reduce the cost of constructing the local surrogate at the current
step.

%% file: appendix.tex
\appendix
\section{Proof of Lemma~\ref{lm:cmt}}
We start with  $L(\-G)$ being uniformly continuous, which means that for any $\epsilon>0$, 
there exists a $\delta>0$ such that,
 for any $|\-G_N-\-G|<\delta$, one has
$|L(\-G_N)-L(\-G)|<\sqrt{\epsilon/2}$. 
On the other hand, $\-G_N(\-y)\rightarrow \-G(\-y)$ in $L^2_q$ as $N\rightarrow\infty$,
implying that $\-G_N(\-y)\rightarrow \-G(\-y)$ in probability as $N\rightarrow\infty$;
therefore, for the given $\epsilon$ and $\delta$, there exists a positive integer $N_o$ such that,
for all $N>N_o$, 
$\P[|\-G_N(\-y)- \-G(\-y)|>\delta]<\epsilon/4$.
Let $\Omega  := \{\-z : |\-G_N(\-y)- \-G(\-y)|<\delta\}$ and
let ${\Omega^\ast}$ be the complement of $\Omega$ in the support of $q$. Then we can
write 
\begin{multline}
\|L(\-G_N(\-y))-L(\-G(\-y))\|^2_{L^2_q} \\=
\int_{\Omega} (L(\-G_N(\-y))-L(\-G(\-y)))^2 q(\-y) d\-y +\int_{{\Omega^\ast}} (L(\-G_N(\-y))-L(\-G(\-y)))^2 q(\-y) d\-y
\label{e:pr}
\end{multline}
The first integral on the right hand side of \eqref{e:pr} is smaller than $\epsilon/2$ by design.
Now recall that the likelihood function $L(\cdot)$ is bounded, and without loss of generality we assume $0<L(\cdot)<1$.
It then follows that the second integral on the right-hand side of \eqref{e:pr} is smaller than $\epsilon/2$ too.
Thus we complete the proof. $\square$